\documentclass[letterpaper,10pt, conference]{ieeeconf}
\IEEEoverridecommandlockouts
\usepackage{amsmath,amssymb,amsthm,epsfig,color,subfigure,empheq,graphicx}
\usepackage[algo2e]{algorithm2e}
\usepackage{enumerate,url,wasysym,balance,enumerate}
\usepackage{algorithm,algpseudocode}	
\usepackage{arydshln}
\usepackage{tikz,balance}
\usetikzlibrary{matrix,decorations.pathreplacing}

\newcommand \bd{\mathbf{d}}

\newcommand \bef{\mathbf{f}}
\newcommand \bg{\mathbf{g}}
\newcommand \bi{\mathbf{i}}

\newcommand \bu{\mathbf{u}}
\newcommand \bv{\mathbf{v}}

\newcommand \bx{\mathbf{u}}
\newcommand \by{\mathbf{y}}

\newcommand \bq{\mathbf{q}}
\newcommand \bA{\mathbf{A}}
\newcommand \bB{\mathbf{B}}

\newcommand \bG{\mathbf{G}}
\newcommand \1{\mathbf{1}}

\newcommand \blambda{{\boldsymbol{\lambda}}}
\newcommand \bpi{{\boldsymbol{\pi}}}

\newcommand \mcG{\mathcal{G}}

\newcommand \mcF{\mathcal{F}}

\newcommand \mcN{\mathcal{N}}

\newcommand \mcE{\mathcal{E}}

\newcommand \mcY{\mathcal{Y}}

\newcommand \bzero{\boldsymbol{0}}

\newcommand \bfeta{\boldsymbol{\eta}}

\renewcommand{\baselinestretch}{0.98}

\DeclareMathOperator{\diag}{dg}

\newtheorem{proposition}{Proposition}
\newtheorem{lemma}{Lemma}

\newtheorem{assumption}{Assumption}

\theoremstyle{remark}\newtheorem{remark}{Remark}

\allowdisplaybreaks

\usepackage{ifthen}
\newboolean{showcomments}
\setboolean{showcomments}{true}
\newcommand{\andrey}[1]{\ifthenelse{\boolean{showcomments}}
	{ \textcolor{blue}{(Andrey says:  #1)}}{}}
\newcommand{\guido}[1]{\ifthenelse{\boolean{showcomments}}
	{ \textcolor{blue}{(Guido says:  #1)}}{}}

\title{\LARGE \bf Ripple-Type Control for Enhancing Resilience\\of Networked Physical Systems}

\author{Manish K. Singh, Guido Cavraro, Andrey Bernstein, and Vassilis Kekatos%
\thanks{Manish K. Singh and Vassilis Kekatos are  with the Bradley Department of ECE, Virginia Tech, Blacksburg, VA 24061 USA (e-mail:
\{manishks,~kekatos\}@vt.edu).
G. Cavraro and A. Bernstein are with the Power Systems Engineering Center, National Renewable Energy Laboratory, Golden, CO 80401 USA (e-mail: name.surname@nrel.gov). Work partially supported by the U.S. National Science Foundation under Grant 1711587. This work was authored by the National Renewable Energy Laboratory, operated by Alliance for Sustainable Energy, LLC, for the U.S. Department of Energy (DOE) under Contract No. DE-AC36-08GO28308. Funding provided by NREL Laboratory Directed Research and Development Program. The views expressed in the article do not necessarily represent the views of the DOE or the U.S. Government. The U.S. Government retains and the publisher, by accepting the article for publication, acknowledges that the U.S. Government retains a nonexclusive, paid-up, irrevocable, worldwide license to publish or reproduce the published form of this work, or allow others to do so, for U.S. Government purposes.
M. K. Singh and G. Cavraro contributed equally to this work.}%
}

\begin{document}
\maketitle
\thispagestyle{empty}
\pagestyle{empty}

\begin{abstract}
Distributed control agents have been advocated as an effective means for improving the resiliency of our physical infrastructures under unexpected events. Purely local control has been shown to be insufficient, centralized optimal resource allocation approaches can be slow. In this context, we put forth a hybrid low-communication saturation-driven protocol for the coordination of control agents that are distributed over a physical system and are allowed to communicate with peers over a ``hotline'' communication network. According to this protocol, agents act on local readings unless their control resources have been depleted, in which case they send a beacon for assistance to peer agents. Our ripple-type scheme triggers communication locally only for the agents with saturated resources and it is proved to converge. Moreover, under a monotonicity assumption on the underlying physical law coupling control outputs to inputs, the devised control is proved to converge to a configuration satisfying safe operational constraints. The assumption is shown to hold for voltage control in electric power systems and pressure control in water distribution networks. Numerical tests corroborate the efficacy of the novel scheme. 
\end{abstract}


%%%%%%%%%%%%%%%%%%%%%%%%%%%%%%%%%%%%%%%%%%%%%%%%%%%%%%%%%

\section{Introduction}\label{sec:intro}
Utility systems, such as power, water, and gas networks, are significant examples of networked systems and are undergoing rapid changes. Power systems experience significant penetration of distributed energy resources and flexible load thus increasing the system volatility. Natural gas-fired generators serve as a fast-acting balancing mechanism for power systems, inadvertently increasing the volatility in gas networks. Rising threats of clean water scarcity have motivated tremendous efforts toward judicious planning for bulk water systems and enhanced monitoring and control for water distribution networks. The increasing occurrence of natural disasters and other (cyber and physical) disruptions undermines the operation of the aforesaid systems.

Classic operation of utility systems is performed via management systems that aim at satisfying consumer demands in a cost-effective manner while meeting the related operational and physical constraints. Such tasks constitute the family of optimal dispatch problems (ODP). ODPs are typically solved at regular intervals based on anticipated demands and network conditions. Stochastic optimization formulations are often leveraged to account for the uncertainty and to ensure reliable operation within the ODP interval.

Although ODP solutions can ensure reliable system operation during normal conditions, they cannot account for the occurrence  of low-probability, high-impact disruptions that might undermine system operation during the time interval between two ODP actions.
To improve system resilience to these events, this paper devises a control mechanisms to ensure the satisfaction of system operational requirements. It draws features from local, distributed, and event-triggered control and is referred to as \emph{ripple-type control}.  In local control rules, agents make decisions based on locally available readings. For example, in~\cite{HaoZhu}, \cite{Xinyang}, power generators control their reactive power output given their local power injection and voltage; however, local schemes have limited efficacy~\cite{NeedComm}. Distributed control strategies, in which agents compute their control action after sharing information with neighbors in a communication network, have a wide spectrum of applications, e.g.,  energy systems~\cite{CavCar2018}, \cite{bernstein2019}, \cite{Colombino} or camera networks~\cite{Bof}. To avoid wasting resources and to communicate only when it is really needed, event-triggered control techniques were advocated in~\cite{Heemels}, \cite{NaLi2019ACM}. Essentially, every agent evaluates locally a \emph{triggering function}, e.g., in a consensus setup, the mismatch between the current state and the state that was last sent to neighbors~\cite{Cortes2016consensus}. When the triggering function takes some specified values, agents communicate and update their control rule.

In the proposed ripple-type control, agents first try to satisfy their local constraints via purely local control. Only when such control efforts reach their maximum limit, an assistance is sought from neighboring agents on a communication graph. The process is continued until the control objectives of every agent are met. The proposed algorithm is \emph{model free} in the sense that it does not require knowledge of the system model parameters. This is an essential property of real-time emergency control due to lack of accurate model information during contingency scenarios~\cite{chen2020}.
%A feasibility problem is first defined based on the critical hard constraints of the original ODP.  A distributed, event-triggered algorithm is then designed to solve the feasibility problem in real time. The algorithm employs a hybrid approach, wherein agents in the network (e.g, loads and generators in a power system) try to satisfy their local constraints first, only using local control effort. Only when the local control effort is at its maximum limit, an assistance is sought from the neighboring agents on a communication graph. 
%To show its applicability in real systems, 
The proposed control scheme is tested on electric and water networks.

%\footnote{\emph{Notation}: Lower- (upper-) case boldface letters denote column vectors (matrices). Sets are represented by calligraphic symbols. Symbol $^\top$ stands for transposition. Inequalities, $\max$, and $\min$ operators are understood element-wise. All-zero and all-one vectors and matrices are represented by $\bzero$ and $\1$; the respective dimensions are deducible from context. \cmb{Symbol $\succ$ indicates all positive real eigenvalues.} The $\diag(\cdot)$ operator on vectors places the vector on the principal diagonal of a matrix. Symbol $\|\cdot\|$ represents the $L_2$ norm.} 

%%%%%%%%%%%%%%%%%% GRID MODELING  %%%%%%%%%%%%%%%%%%%%%
\section{System Modeling}
\label{sec:net_modeling}
Consider a networked system modeled by an undirected graph $\mcG=(\mcN,\mcE)$. The set $\mcN$ is a collection of $N$ nodes hosting controllable agents, and vector $\bx\in\mathbb{R}^N$ represents their \emph{control inputs}\footnote{\emph{Notation}: Lower- (upper-) case boldface letters denote column vectors (matrices). Sets are represented by calligraphic symbols. Symbol $^\top$ stands for transposition. Inequalities, $\max$, and $\min$ operators are understood element-wise. All-zero and all-one vectors and matrices are represented by $\bzero$ and $\1$; the respective dimensions are deducible from context. Symbol $\succ$ indicates all positive real eigenvalues. The $\diag(\cdot)$ operator on vectors places the vector on the principal diagonal of a matrix. Symbol $\|\cdot\|$ represents the $L_2$ norm.}.
A subset of agents comprising set $\mcY\subset \mcN$ of cardinality $M$ is assumed to be collecting noiseless scalar observations of the local \emph{states or outputs} stacked in $\by\in\mathbb{R}^M$. The entries of $\by$ are to be regulated within a desired range. Given control $\bx$, the system has a locally unique output $\by$ determined by a mapping $\textrm{F}: \mathbb{R}^N\rightarrow \mathbb{R}^M$.
Heed that the mapping $\textrm{F}$ depends implicitly on other uncontrollable system inputs that are not part of $\bu$. Moreover, this mapping might not have an explicit form. We consider cyber-physical systems in which the related mapping $\textrm{F}$ adheres to the following property.

\begin{assumption}\label{as:1}
Mapping $\by=\textrm{F}(\bu)$ satisfies $\nabla_{\bu} \textrm{F}(\bu)\geq \bzero~\forall~\bu$.
\end{assumption}

Albeit seemingly restrictive, the postulated monotonicity assumption holds for several physical systems abiding by a dissipative flow law, such as natural gas and water networks~\cite{Vuffray2015monotonicity}. We will now establish the validity of Assumption~\ref{as:1} for power systems and water networks. Then, Section~\ref{sec:prob_form} formally introduces the problem setup, and Section~\ref{sec:solution} presents the proposed control scheme and the associated claims.

\emph{Power Systems.} An electric power network can be modeled by a graph $\mcG=(\mcN,\mcE)$, with $\mcE$ capturing transmission lines. The node or bus set $\mcN$ can be partitioned into generator (PV) buses $\mcN_G$ and load (PQ) buses $\mcN_L$. Generators can control their active power injection and voltage magnitude. For loads, complex power injections can be adjustable or fixed, but they are largely independent of voltages. Let $(v_n,q_n)$ denote the voltage magnitude and reactive power injection at bus $n\in\mcN$. Vectors $(\bv,\bq)$ collect $\{(v_n,q_n)\}_{n\in\mcN}$.

Consider the task of maintaining the load voltages above given limits by controlling the reactive power injections at the load buses and the voltages at the generator buses. Loads can partially control their reactive injections due to inverters, capacitor banks, and  flexible AC transmission systems (FACTS). To study the task, we build on the widely adopted approximate model~\cite{Porco}:
% Transmission lines are often approximated as lossless, since their resistances are much smaller than their reactances. Thanks to this lossless approximation and the assumption of small voltage angle differences across neighboring buses, voltage angles (resp. magnitudes) are rather insensitive to variations in reactive (resp. active) power injections. To develop voltage control algorithms, the focus here is on the $\bq/\bv$ system . Vectors $\bv$ and $\bq$ adhere to the approximate model
\begin{equation}\label{eq:qdef}
    \bq=\diag(\bv)\bB\bv.
\end{equation}
where $\bB \in \mathbb R^{N\times N}$ is a weighted Laplacian matrix of $\mcG$: Its $(m,n)$-th entry $B_{mn}<0$ equals the negative susceptance of line $(m,n)\in\mcE$; $B_{mn}=0$ if $(m,n)\notin\mcE$; and $B_{mm}=-\sum_{n\neq m}B_{mn}>0$ for its diagonal entries. Model~\eqref{eq:qdef} has been derived from the AC power flow equations after ignoring transmission line resistances and assuming small voltage angle differences across neighboring buses. Partitioning $\bv$ and $\bq$ into generator and load buses, rewrite~\eqref{eq:qdef} as:
\begin{equation}
    \begin{bmatrix}\bq_G\\\bq_L\end{bmatrix}=\begin{bmatrix}\diag(\bv_G)&\bzero\\\bzero&\diag(\bv_L)\end{bmatrix}\begin{bmatrix}\bB_{GG}&\bB_{LG}^\top\\\bB_{LG}&\bB_{LL}\end{bmatrix}\begin{bmatrix}\bv_G\\\bv_L\end{bmatrix}\label{eq:qdef2}
\end{equation}
where $\bB$ has been partitioned accordingly.
%Equation~\eqref{eq:vq} instantiates~\eqref{eq:model} for power systems.

%On the other hand, load voltages and the reactive power output of generators are uncontrolled quantities and so
%
%\begin{equation*}%\label{eq:xdef}
%\phi_n:=\left\{\begin{array}{ll}
%v_n&,~ n\in\mcN_L\\
%q_n &,~ n\in\mcN_S
%\end{array}.\right.
%\end{equation*}
%
%Based on the former reasoning, we set
%
%\begin{equation}
%\bx = [\bv_G^\top ~~ \bq_L^\top]^\top,\quad \bphi = [\bq_G^\top~~\bv_L^\top]^\top.
%\label{eq:x&phi}
%\end{equation}
%
%Given $\bx$, vector $\bphi$ can be found by solving~\eqref{eq:qdef2}. 

For the control task at hand, the control variable per bus $n$ depends on the type of the bus as:
\begin{equation*}%\label{eq:xdef}
u_n:=\left\{\begin{array}{ll}
v_n &,~ n\in\mcN_G\\
q_n&,~ n\in\mcN_L
\end{array}.\right.
\end{equation*}
We would like to adjust input $\bu=[\bq_L^\top~~\bv_G^\top]^\top$ to control output $\by=\bv_L$. To validate Assumption~\ref{as:1}, the ensuing result studies the mapping $\by=\textrm{F}(\bu)$ for the case of voltage control.

\begin{proposition}\label{prop:monotonicPower}
For control input $\bu=[\bq_L^\top~~\bv_G^\top]^\top$ and output $\by=\bv_L$, the mapping $\by=\textrm{F}(\bu)$ satisfies Assumption~\ref{as:1} if $\diag(\bg_L)+\bB_{LL}\succ0$ where $\bg_L:=\left[\diag(\bv_L)\right]^{-2}\bq_L$.
\end{proposition}

\begin{proof}
Adopting the implicit differentiation approach, apply the differential operator on the second block of~\eqref{eq:qdef2}:
\begin{equation}\label{eq:qdiff}
\partial \bq_L=\left(\diag(\bi_L)+\diag(\bv_L)\bB_{LL}\right) \partial \bv_L +\diag(\bv_L)\bB_{LG}\;\partial \bv_G
\end{equation}
where $\bi_L:=\bB_{LG}\bv_G+\bB_{LL}\bv_L$. Because $(\bq_L,\bv_G)$ are independent, it holds $\nabla_{\bq_L} \bv_G=\bzero$ so that \eqref{eq:qdiff} yields:
\begin{equation}\label{eq:dql}
    % \nabla_{\bq_L} \bv_L=\frac{\partial\bv_L}{\partial \bq_L}=\left(\diag(\bi)+\diag(\bv_L)\bB_{LL}\right)^{-1}.
    \nabla_{\bq_L} \bv_L=\left(\diag(\bi_L)+\diag(\bv_L)\bB_{LL}\right)^{-1}.
\end{equation}
And because $\nabla_{\bv_G} \bq_L=\bzero$, equation \eqref{eq:qdiff} also provides:
\begin{equation}\label{eq:dvg}
    % \nabla_{\bv_G} \bv_L=\frac{\partial \bv_L}{\partial \bv_G}=-\left(\diag(\bi)+\diag(\bv_L)\bB_{LL}\right)^{-1}\diag(\bv_L)\bB_{LG}.
    \nabla_{\bv_G} \bv_L=-\left(\diag(\bi_L)+\diag(\bv_L)\bB_{LL}\right)^{-1}\diag(\bv_L)\bB_{LG}.
\end{equation}

% where $\bi=\bB_{LG}\bv_G+\bB_{LL}\bv_L$ represents the current injection at load buses, and where we used the fact that $\bv_G$ and  $\bq_L$ are independent variables.

For the Jacobian matrices $\nabla_{\bq_L} \bv_L$ and $\nabla_{\bv_G} \bv_L$ to have nonnegative entries, it suffices to show that the inverse of $\bG:=\diag(\bi_L)+\diag(\bv_L)\bB_{LL}$ has nonnegative entries. This is because $\bv_L> \bzero$ and $\bB_{LG}\leq \bzero$. 

To establish $\bG^{-1}\geq \bzero$, note that the off-diagonal entries of $\bG$ are nonpositive; hence, proving $\bG\succ0$ would make $\bG$ an M-matrix so that~$\bG^{-1}\leq\bzero$. The assumption $\diag(\bg_L)+\bB_{LL}\succ0$ stated in this proposition ensures $\bG\succ0$, as we show next. From \eqref{eq:qdiff}, it follows that $\bq_L=\diag(\bv_L)\bi_L$. Substituting $\bi_L=\diag(\bv_L)^{-1}\bq_L$ in the stated condition yields:
\begin{align*}
&\diag(\bi_L)\diag(\bv_L)^{-1}+\bB_{LL}\succ 0\implies\\
&\diag(\bi_L)+\diag(\bv_L)^{1/2}\bB_{LL}\diag(\bv_L)^{1/2}\succ 0\implies\\
&\diag(\bi_L)+\diag(\bv_L)\bB_{LL}=\bG\succ0
\end{align*}
where the first transition follows from Sylvester's law of inertia for congruent matrices, and the second from the similarity transformation involved.
\end{proof}

Proposition~\ref{prop:monotonicPower} identifies the conditions that ensure Assumption~\ref{as:1} holds for the grid model in~\eqref{eq:qdef}. Numerically verifying $\diag(\bg_L)+\bB_{LL}\succ0$ for benchmark systems revealed an interesting observation, discussed next. For the IEEE 5-39-118-bus systems, we scaled up the nominal $\bq_L$ by a scalar until the power flow solver MATPOWER~\cite{MATPOWER} failed to converge. At each step, we also noted the eigenvalues of $\diag(\bg_L)+\bB_{LL}$. Interestingly, the minimum eigenvalue kept decreasing for increasing $\bq_L$, but it remained positive until the last successful power flow instance for all networks. This indicates a relation between $\diag(\bg_L)+\bB_{LL}\succ0$ and the solvability of the AC power flow equations, but its analytical investigation goes beyond the scope of this work.

%Proposition~\ref{prop:monotonicPower} shows that Assumption~\ref{as:1} holds for the grid model in~\eqref{eq:qdef} under the identified conditions. Numerically verifying $\diag(\bg_L)+\bB_{LL}\succ0$ for real-world systems led us to an interesting observation described next. For the IEEE 5-, 39-, and 118-bus systems, we scaled up the nominal $\bq_L$ by a scalar until the power flow solver MATPOWER~\cite{MATPOWER} failed to converge. At every step, we also observed the eigenvalues of $\diag(\bg_L)+\bB_{LL}$. Interestingly, the minimum eigenvalue kept decreasing for increasing $\bq_L$, but remained positive up until the last successful power flow instance for all networks. This indicates a relation between $\diag(\bg_L)+\bB_{LL}\succ0$ and the solvability of the AC power flow equations, but its analytical investigation goes beyond the scope of this work.

\emph{Water Distribution Systems (WDS).}  
A WDS can be modeled by an undirected graph $\mcG=(\mcN,\mcE)$, where the nodes in $\mcN$ correspond to water reservoirs, tanks, junctions and consumers. If $d_n$ denotes the rate of water injection at node $n$, then $d_n\geq0$ for reservoirs; $d_n\leq0$ for water consumers; tanks might be filling or emptying; and $d_n=0$ for junctions. Nodes are connected by edges in $\mcE$, which represent pipelines, pumps, and valves. Let $\sigma_{mn}$ denote the rate of water flowing from node $m$ to $n$ over edge $(m,n)$. It also holds $\sigma_{mn}=-\sigma_{nm}$. Flow conservation at node $n$ dictates:
\begin{equation}\label{eq:WF1}
    d_n=\sum_{m:(n,m)\in\mcE}\sigma_{nm}.
\end{equation}
The relation between water flow $\sigma_{mn}$ across edge $(m,n)$ and the pressures $\pi_m$ and $\pi_n$ at nodes $m$ and $n$ takes the form 
\begin{equation}\label{eq:WF2}
    \pi_m-\pi_n=\rho_{mn}(\sigma_{mn}).
\end{equation}

The operation of a WDS involves serving water demands while maintaining nodal pressures within desirable levels. The task of pressure control can include continuous-valued variables, such as reservoir and tank output pressures, as well as water injections at different nodes. Pressure control can also involve binary variables capturing the on/off status of fixed-speed pumps and valves~\cite{singh2018optimal}. These continuous and binary control variables are often determined by periodically solving ODPs; see e.g.,~\cite{singh2018optimal}. Assuming the binary variables to be fixed between two ODP instances, we focus on the continuous-valued variables. Partition $\mcN$ into the subset $\mcN_S$ of reservoirs and tanks, and the subset $\mcN_L$ of loads. For nodes in $\mcN_S$, pressures are controllable. For nodes in $\mcN_L$, demands are controllable; inelastic demands can be modeled with lower and upper limits coinciding; thus, the control variables now are:
\begin{equation*}
u_n:=\left\{\begin{array}{ll}
\pi_n&,~ n\in\mcN_S\\
d_n &,~ n\in\mcN_L
\end{array}.\right.
\end{equation*}
Given $\bx$, the water injections for the nodes in $\mcN_S$ and the pressures at nodes in $\mcN_L$ are determined by the water flow equations~\eqref{eq:WF1}--\eqref{eq:WF2}. The nodal pressures over $\mcN_L$ need to be maintained at stipulated levels; thus they constitute vector $\by$. For the aforementioned assignments of $\bx$ and $\by$, the mapping $\textrm{F}$ is implicitly defined by \eqref{eq:WF1}--\eqref{eq:WF2}. Reference~\cite{SinghKekatosWF19} guarantees that $\textrm{F}$ maps a given $\bx$ to a unique $\by$.

To verify the validity of Assumption~\ref{as:1}, we use a monotonicity result for dissipative flow networks from~\cite{Vuffray2015monotonicity}. 
To qualify as a dissipative flow network per~\cite{Vuffray2015monotonicity}, the functions $\rho_{mn}(\cdot)$ in \eqref{eq:WF2} should be nondecreasing and continuous. Both these requirements hold true for edges in a WDS. Specifically, for a pipe $(m,n)\in\mcE$, the function $\rho_{mn}$ models the pressure drop due to friction described by the Darcy-Weisbach or the Hazen-Williams laws~\cite{SinghKekatosWF19}. Both are non-decreasing and continuous. For pumps and valves, pressures $\pi_m$ and $\pi_n$ relate to flow $\sigma_{mn}$ via nondecreasing empirical laws as well~\cite{CohenQHmodel}; therefore, WDS fall within the purview of dissipative networks, and the ensuing result applies~\cite{Vuffray2015monotonicity}. 

\begin{lemma}\cite[Corollary~4]{Vuffray2015monotonicity}\label{le:monotonicity}
Let $\bpi$ and $\bpi'$ be $N$-length pressure vectors satisfying the water flow equations~\eqref{eq:WF1}-\eqref{eq:WF2} for demand vectors $\bd$ and $\bd'$. If $\pi_n\geq\pi_n'$ for all $n\in\mcN_S$, and $d_n\geq d_n'$ for all $n\in\mcN_L$, then $\pi_n\geq\pi_n'$ for all $n\in\mcN_L$.
\end{lemma}

In terms of the assigned vectors $\bx$ and $\by$, Lemma~\ref{le:monotonicity} states that controls $\bx\geq\bx'$ result in $\by\geq\by'$; thus, considering infinitesimal changes at node pairs $m$ and $n$, Lemma~\ref{le:monotonicity} translates to $\frac{\partial y_m}{\partial u_n}\geq 0$, hence satisfying Assumption~\ref{as:1}.

\section{Problem Formulation}
\label{sec:prob_form}
Operators manage the networked system by computing periodically the control set points for agents to implement. Such set points are typically the solution of an ODP
\begin{subequations}
\begin{align}
\bx^* \in \arg\min_{\bx}~&~c(\bx,\by)\tag{P1}\label{eq:ODP}\\
\mathrm{s.to}~& ~\by=\textrm{F}(\bx)\label{seq:NF}\\
%\bg(\bx,\by)=\bzero \label{seq:NF}\\
~&~ \mathbf{h}(\bx,\by)\leq0, \label{seq:NFlimits}\\
~&~ \underline{\bx}\leq\bx\leq\bar{\bx}~\label{seq:xlim}\\
~&~ \by \geq \underline{\by} ~\label{seq:ylim}
\end{align}
\end{subequations}
where $c(\bx,\by)$ is a cost function depending on the system inputs and outputs; the mapping in~\eqref{seq:NF} corresponds to the physical laws governing the networked system; and the inequality constraints are grouped into two categories:
\begin{itemize}
\item Constraint~\eqref{seq:NFlimits} captures requirements that are important to efficiently operate the system, but can in principle be safely violated (especially for a short time). These will be referred to as \emph{soft constraints}.
\item Constraints~\eqref{seq:xlim} and \eqref{seq:ylim} impose limitations on input and output variables. If violated, they can lead to system failure. These constraints will be referred to as \emph{hard constraints}.    
\end{itemize}

Typically, given the set of parameters defining~$\textrm{F}$, a central dispatcher solves~\eqref{eq:ODP} and communicates optimal control set points $\bx^*$ to agents. Ideally, this process shall be repeated every time there is a change in the system model that modifies the underlying definition of $\textrm{F}$, such as an abrupt load change or line tripping for the case of power systems. Constrained by communication and computational resources, however, problem \eqref{eq:ODP} is solved only at finite time intervals. As a consequence, the set points $\bx^*$ can become obsolete or even result in network constraint violations.

Such limitations motivate the design of mechanisms to at least ensure that some important operational requirements are met between two consecutive centralized dispatch actions, i.e., to make the control vector $\bx$ belong to the feasible set: 
$$\mcF = \{ \bx: ~ \underline{\bx}\leq\bx\leq\bar{\bx},~\by=\textrm{F}'(\bx), \by\geq\underline{\by}\}$$ 
in which $\textrm{F}'$ is the  input-output mapping defined by the system model after a disruptive event. Set $\mcF$ considers only the hard constraints of~\eqref{eq:ODP}, namely, \eqref{seq:xlim} and \eqref{seq:ylim}. 
We next put forth an algorithm for steering $\bx$ to $\mcF$.
In the following, we will assume that the set $\mcF$ is non-empty, or $\mcF \neq \emptyset$.
%\textcolor{blue}{
%\begin{remark}
%Upper limits on outputs of the type $\by \leq \overline \by$ can be incorporated in~\eqref{seq:NFlimits}. However, the distributed algorithm devised in this work, which promotes cooperation among agents to meet~\eqref{seq:xlim} and~\eqref{seq:ylim}, does not aim at meeting such constraints. The main reason is that local control rules, where agents do not cooperate, are often enough to steer the system to a configuration meeting $\by \leq \overline \by$; e.g., for a power system case, see~\cite{Energies19}, in which overvoltages were resolved by means of local integral controllers. 
%\end{remark}
%}

\begin{remark}
For power systems, constraint~\eqref{seq:NFlimits} models line flow limits, which can be considered soft constraints because their violation could be briefly tolerated. Limits on the power or voltage output of generators are captured by~\eqref{seq:xlim}. Inequality~\eqref{seq:ylim} models load undervoltage constraints. Avoiding dangerously low voltages is important to prevent a voltage collapse. For WDS, the inequalities~\eqref{seq:NFlimits} can represent the water flow limits on pipes and pumps. Although flow limits shall normally be explicitly enforced, in practice low-pressure limits are triggered before the flow limits~\cite{bragalli2015mathprog}; thus, edge flow can be characterized as soft over short time intervals. Constraints~\eqref{seq:xlim} represent the limits on pressure at water sources and demands by consumers. These limits cannot be violated physically because they are posed by the actual capacity of the WDS components and available demands. The inequalities in~\eqref{seq:ylim} model the minimum pressure requirements at consumer nodes. The latter requirements must be adhered to at all times as low pressures can cause service failures and equipment malfunction at the consumer end. 
\end{remark}

\section{Ripple-Type Network Control}\label{sec:solution}

Our algorithm has the ensuing key attributes:
\renewcommand{\labelenumi}{\emph{A\arabic{enumi})}}
\begin{enumerate}
\item Agent $n\in \mcY$ measures $y_n$ and controls $u_n$ locally. Upon a violation of~\eqref{seq:ylim}, local control resources are used first.
\item Agent $n$ transmits communication signals to a few peer agents over a communication network only if $u_n=\bar{u}_n$, i.e., only when local control resources reach their maximum limit and assistance is sought from neighboring nodes on a communication graph.
\item The control scheme is agnostic to the physical  system parameters, i.e., it is a model-free approach that does not require explicit knowledge of operator $\textrm{F}'$.
\end{enumerate}

The communication network is modeled as an undirected graph $\mcG_c = (\mcN, \mcE_c)$, in which the communication links $\mcE_c$ do not necessarily coincide with the physical connections among agents. Graph $\mcG_c$ is henceforth assumed to be connected, and we denote its adjacency matrix as $\bA$, where $A_{mn} = 1$ if there is a direct communication link between nodes $m$ and $n$, i.e., $(m,n) \in \mcE_c$; and $A_{mn} = 0$ otherwise. Let us also introduce function $\bef:\mathbb{R}^N \rightarrow \mathbb{R}^N_+$ defined entry-wise as:
\begin{equation}\label{eq:fdef}
f_n(\bx):=\left\{\begin{array}{ll}
\underline{y}_n-y_n(\bx),&\underline{y}_n\geq y_n(\bx), n \in \mcY \\
0,&\text{otherwise}.
\end{array}\right.
\end{equation}

Assume that, at time $t=0$, a violation of~\eqref{seq:ylim} occurs as a consequence of a model change. 
%Let $\bx(0)$ and $\by(0)$ be the initial control and observed variables,
Let $\bx(0)$ be the initial control variables, and introduce an auxiliary vector $\blambda\in\mathbb{R}^N$, which is initialized as $\blambda(0)=\bzero$. At subsequent times $t\geq1$, the control scheme proceeds in four steps, as delineated next.

\emph{Step 1:} Agents compute $\bef(\bx(t))$ according to \eqref{eq:fdef}. The entries of $\bef(\bx(t))$ are strictly positive if the associated nodes in $\mcY$ experience a violation of~\eqref{seq:ylim}; and zero, otherwise.

\emph{Step 2:} A target set point is computed as:
\begin{equation}\label{eq:xhat}
    \hat{\bx}(t+1)=\bx(t)+\diag(\bfeta_1)\bef(\bx(t)) + \diag(\bfeta_2) \bA \blambda(t)
\end{equation}
for positive $\bfeta_1$ and $\bfeta_2$. Note that for node $n$, the target $\hat{u}_n(t)$ is computed using the local reading $y_n(t)$ and the entries of $\blambda$ sent from its peers (neighbor nodes of node $n$ on $\mcG_c$).

\emph{Step 3:} Agents compute the auxiliary vector $\blambda$ as
\begin{equation}\label{eq:lambda-project}
    \blambda(t+1)=\max\{\bzero,\diag(\bfeta_3)(\hat{\bx}(t+1)-\bar{\bx})\}.
\end{equation}
for a positive $\bfeta_3$. Vector $\blambda$ serves as a beacon for assistance that is communicated across peer nodes.

\emph{Step 4:} The target set point is projected to the feasible range: %of $\bx$ as 
\begin{equation}\label{eq:x-project}
    \bx(t+1)=\min\{\hat{\bx}(t+1),\bar{\bx}\},
\end{equation}
and is physically implemented.

%\begin{algorithm}[t]
%	\caption{Decentralized approach for solving \eqref{eq:P1}}\label{algo:1}
%	\SetAlgoLined
%	\SetKwInOut{Input}{Input}
%	\SetKwInOut{Output}{Output}
%	\SetKwInOut{Return}{Return}
%	\SetKwInOut{Initialize}{Initialize}
%	\Input{initial $\bx(0)$ and $\by(0)$; limits $\underline{\bx},~\bar{\bx}$, and $\underline{\by}$; positive constants $\bfeta_1,~\bfeta_2$, and $\bfeta_3$; and matrix $\bL$}
%	\Initialize{Set $\blambda(0)\leftarrow\bzero$, %$\bu_{-1}=\boldsymbol{\infty}$, 
%	and $t\leftarrow0$}
%		\emph{S1:} Compute $\bef(\by(t))$.\\
%		\emph{S2:} $\hat{\bx}(t+1)=\bx(t)+\diag(\bfeta_1)\bef(\by(t))-\diag(\bfeta_2)\bL\blambda(t)$\\
%		\emph{S3:} Update control $\bx(t+1)=\max\{\hat{\bx}(t+1),~\bar{\bx}\}$\\
%		\emph{S4:} Compute $\blambda(t+1)=\max\{\bzero,~\diag(\bfeta_3)(\hat{\bx}(t+1)-\bar{\bx})\}$\\
%		\emph{S5:} Measure and update $\by(t+1)$\\
%		$t\leftarrow t+1$
%	%\eIf{$\psi(0)(\lambda)>\psi(0)$}{
			%Set $\underline{\lambda} \leftarrow \lambda,~\lambda\leftarrow\frac{\underline{\lambda}+\overline{\lambda}}{2}$
		%}{
			%Set $\overline{\lambda} \leftarrow \lambda,~\lambda\leftarrow\frac{\underline{\lambda}+\overline{\lambda}}{2}$
		%}
%\Return{flows $\bphi_\mcC=\bphi_\mcC'+\lambda\bn_\mcC$ and pressures}
%\end{algorithm}

To establish the effectiveness, Proposition~\ref{prop:converge} proves that the proposed scheme reaches an equilibrium point, and Proposition~\ref{prop:equilibrium} proves that this equilibrium belongs to $\mcF$. The proofs for these results are provided in the Appendix.

\begin{proposition}\label{prop:converge}
Given any $\bx(0)$, the sequence $\{\bx(t)\}$ converges asymptotically.
\end{proposition}

\begin{proposition}\label{prop:equilibrium}
Let Assumptions~\ref{as:1} hold, let $\mcF \neq \emptyset$, and:
\begin{equation}
    \|\diag(\bfeta_2)\diag(\bfeta_3)\bA\|<1.
    \label{eq:cond_conv}
\end{equation}
A pair $(\bx,~\blambda)$ is an equilibrium for the proposed scheme if and only if $\bx$ belongs to $\mcF$ and $\blambda=\bzero$. 
\end{proposition}

The novel control scheme satisfies attribute \emph{A1)} by design. Moreover, for a node $n$ with target set points $\hat{u}_n(t+1)$ within the local control limit $\bar{u}_n$, the corresponding entry of $\blambda(t+1)$ is zero; thus, the computation of $\hat{u}_m$ from \eqref{eq:xhat} for nodes $m$ that are neighbors of $n$ requires no communication from node $n$, hence fulfilling \emph{A2)}. The scheme also meets \emph{A3)} because it is agnostic to the physical network topology and/or demands.

\begin{remark}
Apparently, parameter $\bfeta_1$ does not influence the control scheme convergence. Indeed,~\eqref{eq:cond_conv} provides a condition only on $\bfeta_2$ and $\bfeta_3$; however, $\bfeta_1$ affects the equilibrium point and the convergence rate, whose analytical quantification is beyond the scope of this work. Optimizing upon the parameters $\bfeta_1-\bfeta_3$, and communication graph $\mcG_c$ constitutes pertinent future research directions.
\end{remark}

\section{Numerical Tests}\label{sec:tests}
The control algorithm was tested for two utility network applications: voltage control in power systems and pressure control in water systems. Parameters $\bfeta_2$ and $\bfeta_3$ were chosen so that~\eqref{eq:cond_conv} holds true. The communication graphs $\mcG_c$ were arbitrarily chosen while ensuring connectivity.

\begin{figure}[t]
    \centering
    \includegraphics[width=0.30\textwidth]{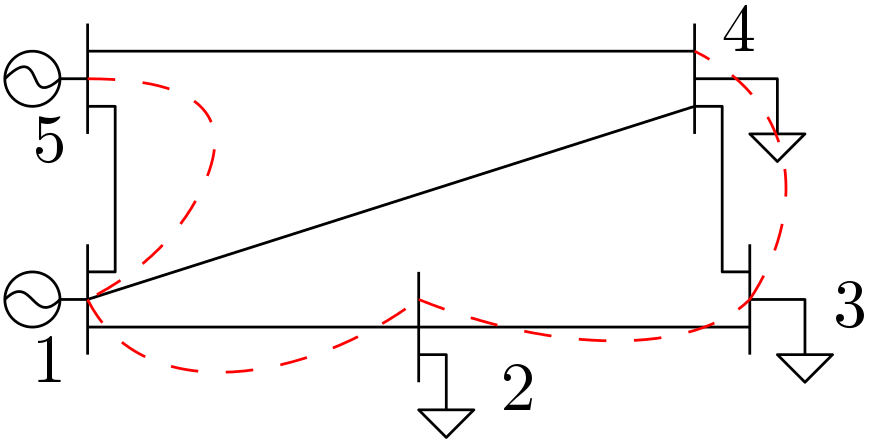}
    \caption{Modified PJM system. Communication links as red dashed lines.}
    \label{fig:testbed}
\end{figure}

\begin{figure}[t]
    \centering
    \includegraphics[scale=0.18]{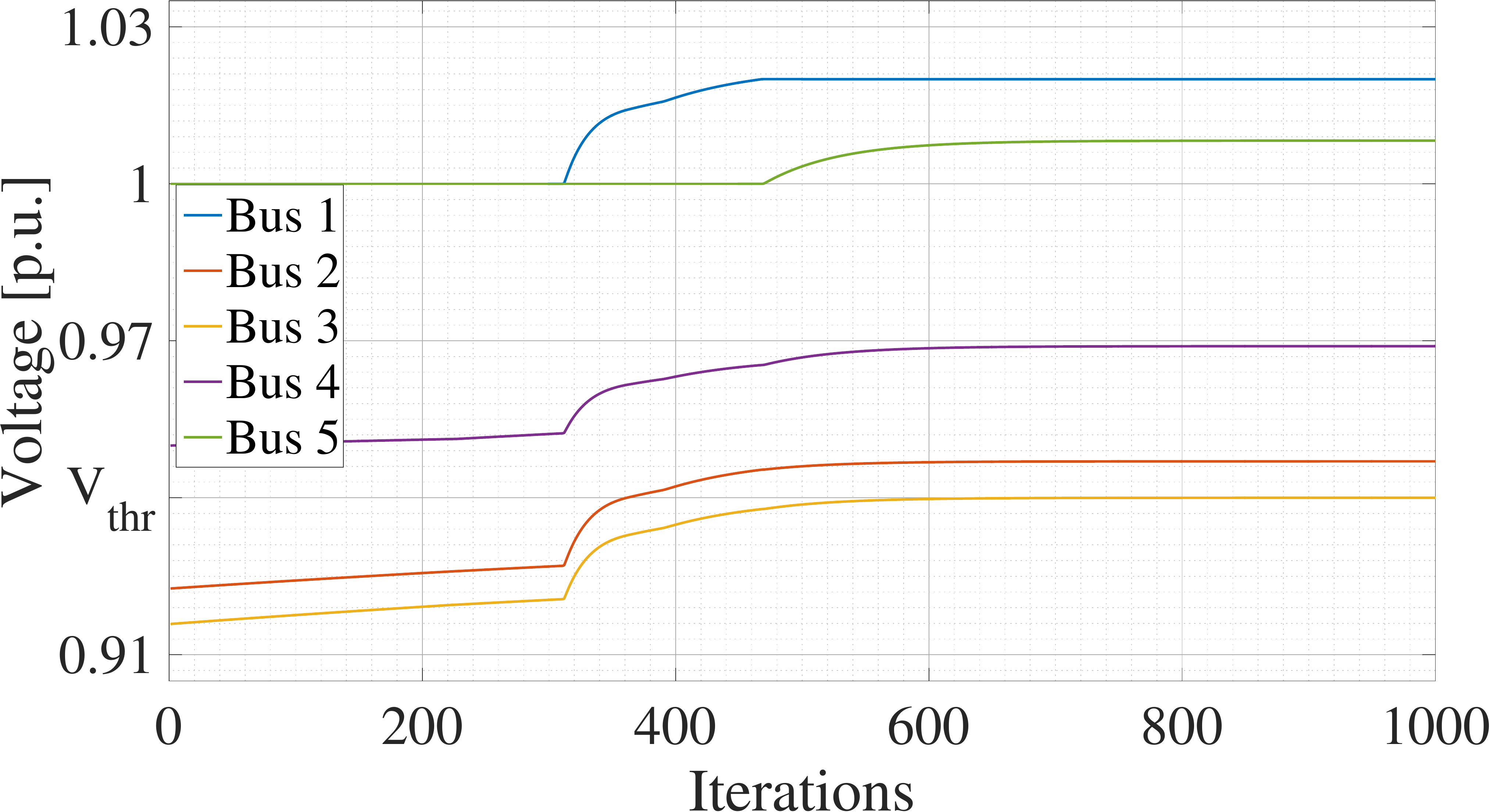}\vspace*{0.4em}
    \includegraphics[scale=0.18]{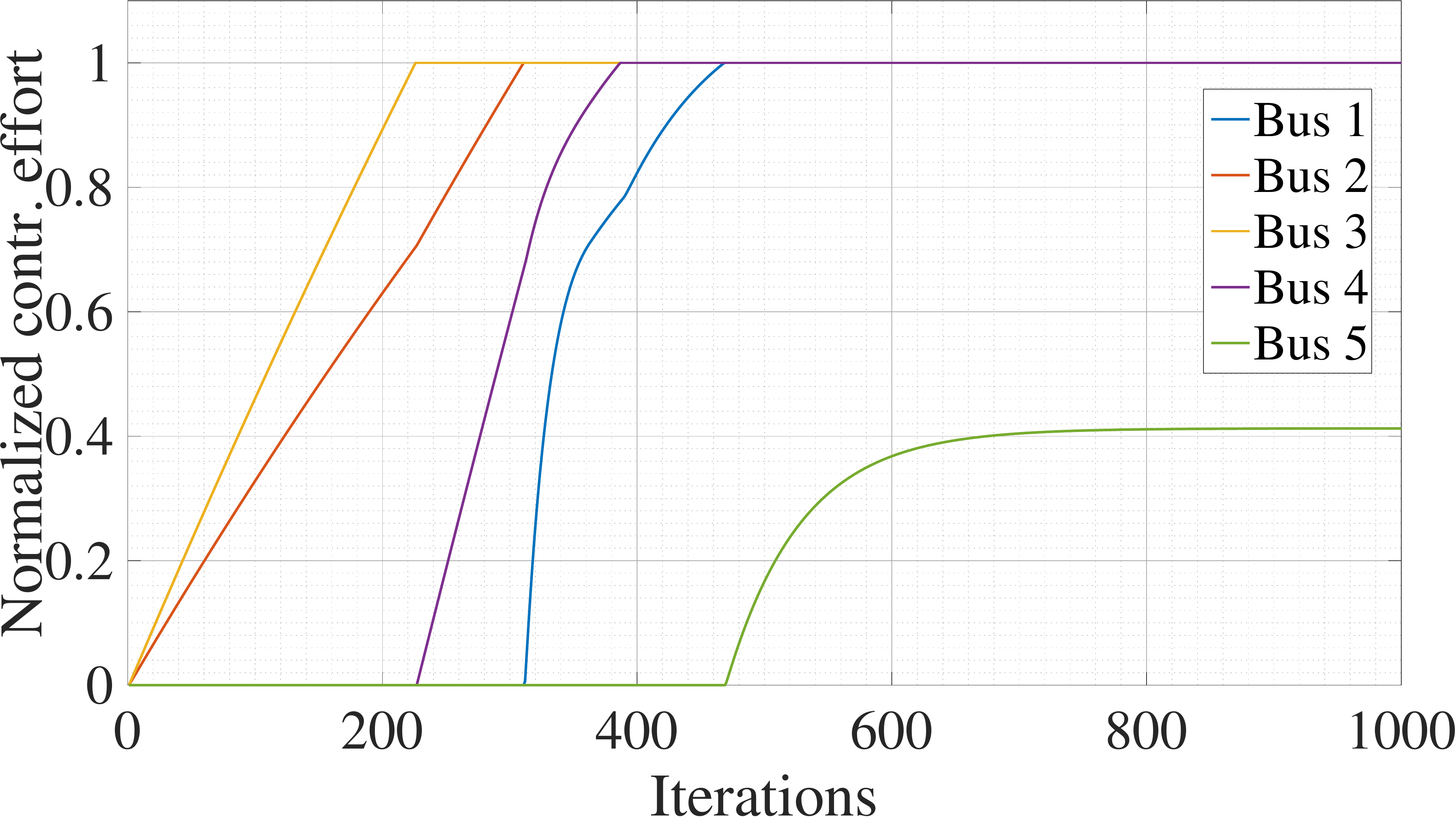}
    \caption{\emph{Top:} Bus voltages. \emph{Bottom:} Normalized control effort $(u_n(t)-u_n(0))/(\bar{u}_n(t)-u_n(0))$ used.}
    \label{fig:PS_results}
\end{figure}

\paragraph{A Power System Test Case}\label{subsec:powertests}
 The benchmark network used for the tests on power systems is a modified version\footnote{In our numerical setup, the generators at Bus 3 and Bus 4 were neglected.} of the PJM 5-bus system~\cite{LiBo} shown in Fig.~\ref{fig:testbed}. Buses 1 and 2 serve as generators, whereas buses 3, 4, and 5 are loads. Generators are initially at $u_n(0)=1$~p.u., allowed to increase their voltage output to 1.02~pu; loads can reduce their power demand by 10~MW, and their voltages are required to be greater than $V_{thr} = 0.94$~p.u., which represents a safe threshold. The top panel of Figure~\ref{fig:PS_results} shows the bus voltage trajectories, whereas the bottom panel plots the ratio of control resources used $(u_n(t)-u_n(0)/(\bar{u}_n(t)-u_n(0))$. After an abnormal event, buses 2 and 3 have voltage less than $V_{thr}$ and start performing the proposed control strategy. Because their local efforts do not manage to regulate the voltages, they seek assistance from their neighbors in the communication network, precisely, Bus 4 (around the 200-th iteration) and Bus 1 (around the 300-th iteration). Finally, when Bus 1 hits its control limits, Bus 5 kicks in (after the 400-th iteration) and is finally able to bring the voltage within the safe interval.

\paragraph{A Water Network Test Case}

\begin{figure}[t]
    \centering
    \includegraphics[width=0.40\textwidth]{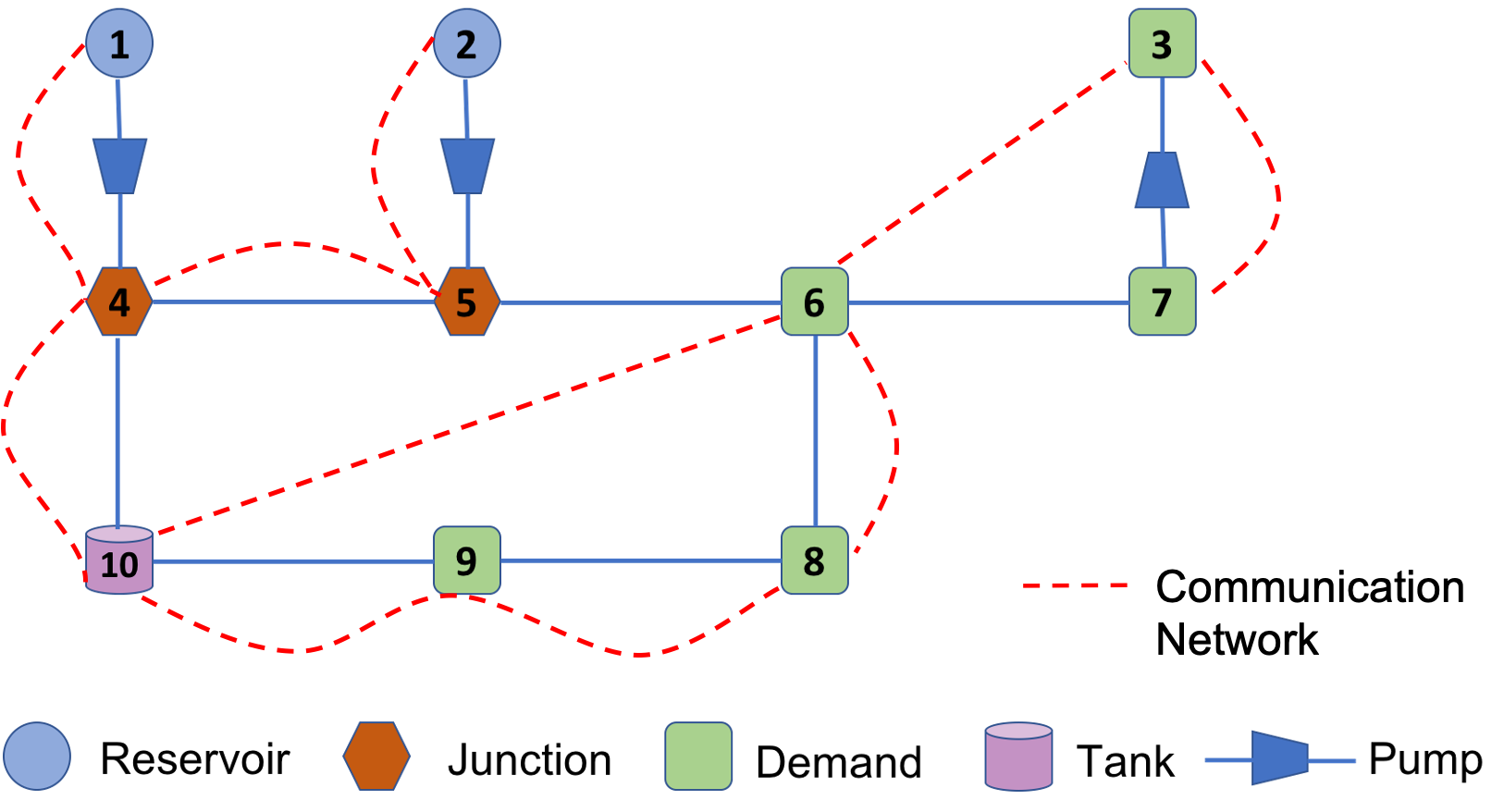}
    \caption{A 10-node benchmark water distribution system.}
    \label{fig:wds}
\end{figure}

\begin{figure}[t]
    \centering
    \includegraphics[scale=0.18]{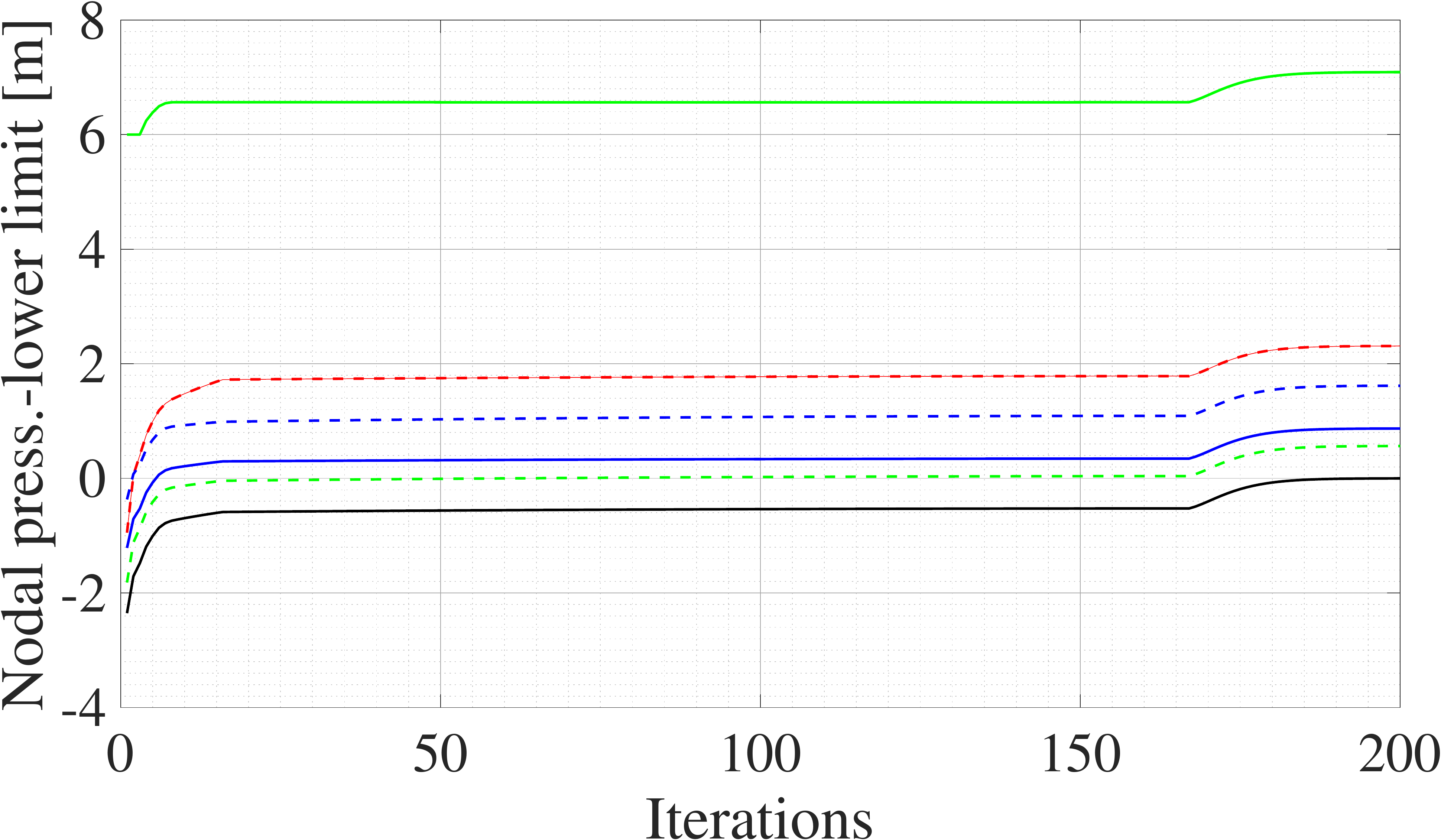}\vspace*{0.4em}
    \includegraphics[scale=0.18]{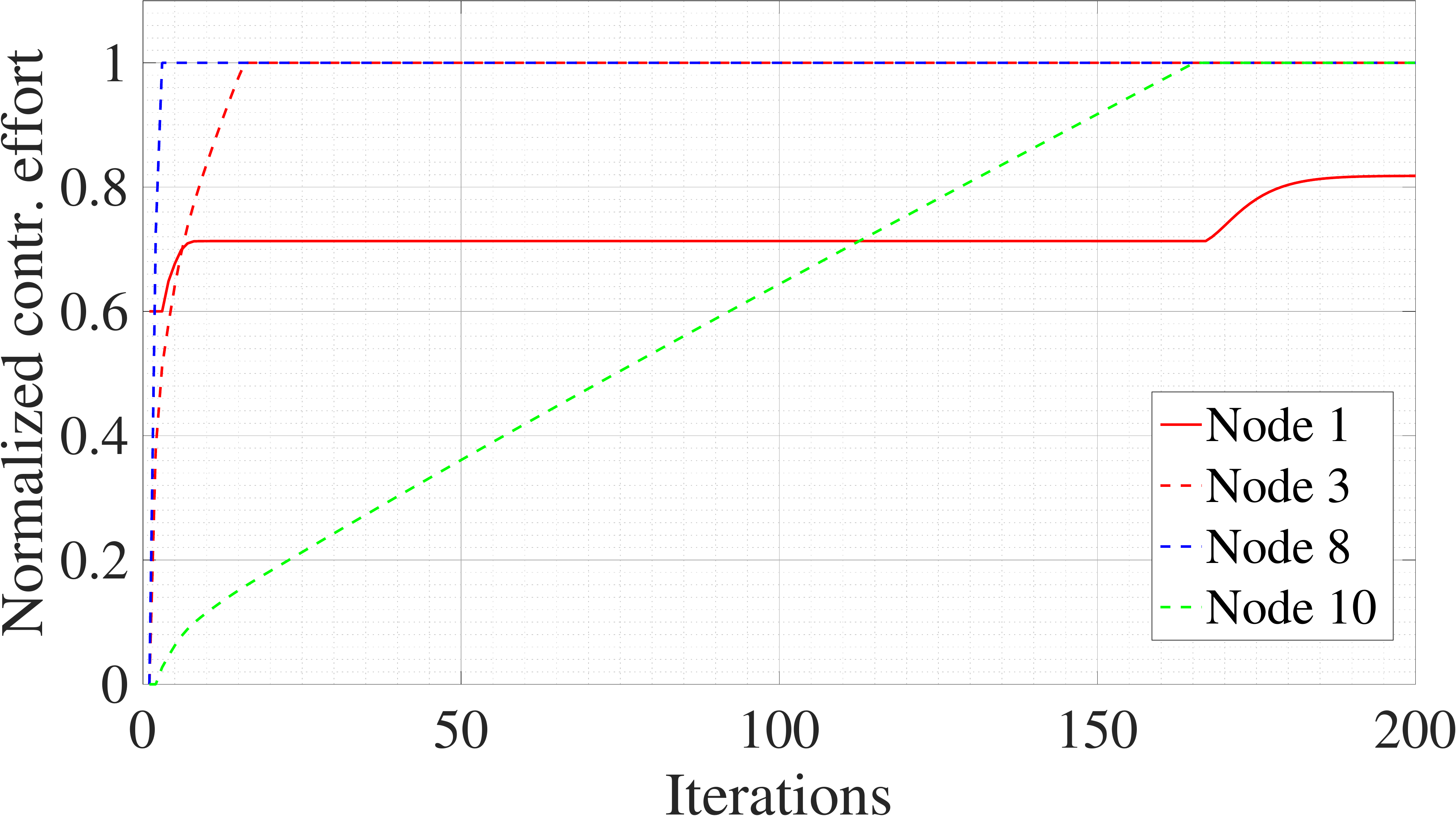}
    \caption{\emph{Top:} Difference of pressures at nodes 3-10 from their lower limits. \emph{Bottom:} Normalized control effort at nodes 1, 3, 8, and 10}
    \label{fig:wdsresults}
\end{figure}

The proposed algorithm was applied the pressure task on the 10-node WDS shown on Fig.~\ref{fig:wds}. % which consists of 2 reservoirs and a tank; 3 fixed-speed pumps; and 7 lossy pipes.
 Pipe dimensions and friction coefficients were taken from~\cite{singh2018optimal}. The minimum pressure requirement $\underline{\pi}_n$ for nodes 3 to 10
is $\{10,7,10,10,5,10,10,10\}$~m. The pumps $(1,2)$, $(2,5)$, and $(7,3)$ are considered to be operating at fixed speeds with constant pressure gains of $10$, $10$, and $5$~m, respectively. All water flow instances were solved using the optimization-based solver of~\cite{SinghKekatosWF19}. The base operating condition involves injection $\bd_0=[380,300,-170,0,0,-220,-200,-150,-140,200]^\top$m$^3/$hr and $\bpi=[3.0,1.8,11.6,13,11.9,10.2,6.6,10.3,10.9,12.9]^\top$m. A disruption was modeled by considering a failure of pump~$(2,5)$, resulting in the unavailability of the Node~2 reservoir and the reservoir at Node~1 supplying $d_1=680$m$^3/$hr. This contingency would result in $\bpi=[3.0,0,9.1,13,8.8,7.7,4.1,8.2,9.6,12.7]^\top$~m, which violates the minimum pressure needed at nodes $3$ and $5-9$.

To study a scenario where not all agents have control capability, only two demand nodes --namely $3$ and $8$-- are allowed to reduce demand by $50$~m$^3/$h. Further, the tank node $10$ has an injection flexibility of $\pm 200$~m$^3/$h and reservoir $1$ has a controllable pressure range of $[0,5]$~m. The performance of the proposed algorithm in restoring the pressures is demonstrated in Fig.~\ref{fig:wdsresults}. The top panel shows the difference $\pi_n-\underline{\pi}_n$ for nodes $n=\{3\dots10\}$. As anticipated, the pressures are nondecreasing and the algorithm succeeds in restoring them above the respective lower limits. The bottom panel plots the normalized nodal control effort. As desired, once all pressures are restored to the desired levels, the algorithm attains an equilibrium, and $\bx$ saturates.

\section{Conclusions}
A ripple-type coordination scheme for the emergency control of networked systems has been put forth. Agents act based on local control rules as long as local resources have not been saturated. Otherwise, they solicit help from peer agents through a ``hotline'' communication network not necessarily coinciding with the physical system graph. The algorithm provably converges to safe operating conditions under an appropriate choice of parameters. Its validity has been illustrated on power and water network examples. Future research directions include enforcing upper limits on the output variables and optimally designing control parameters. 

%%%%%%%%%%%%%%%%%% Appendix %%%%%%%%%%%%%%%%%%%%%
\appendix
%\subsection{Proof of Proposition~\ref{prop:converge}}
%\begin{proof}[Proof of Proposition~\ref{prop:converge}]
\emph{Proof of Proposition~\ref{prop:converge}} Owing to the projection in~\eqref{eq:x-project}, it is evident that $\bu(t)\leq\bar{\bu}$ for all $t$; thus, proving a non-decreasing property $\bu(t)\leq\bu(t+1)$ for all $t$ is sufficient for establishing convergence of the sequence $\{\bu(t)\}$. Consider an arbitrary node $n$ and time $t$. When $u_n(t+1)=\bar{u}_n$ we trivially have $u_n(t) \leq u_n(t+1)$. When $u_n(t+1) < \bar{u}_n$, \emph{Step 4} yields $\hat u_n(t+1) = u_n(t+1)$ and \eqref{eq:xhat} implies:
\begin{equation*}
    u_n(t+1) = u_n(t)+\eta_{1,n}f_n(\bu(t))+\eta_{2,n}\sum_{m=1}^N A_{nm}\lambda_m(t)\geq  u_n(t)
\end{equation*}
because $f_n(\bu(t))$ and $\blambda_m(t)$ are nonnegative; $\eta_{1,n}$ and $\eta_{2,n}$ are positive; and matrix $\bA$ has nonnegative entries.
%\end{proof}

%\begin{proof}[Proof of Proposition~\ref{prop:equilibrium}]
\emph{Proof of Proposition~\ref{prop:equilibrium}}
Assume $\bu \in \mcF$ and $\blambda=\bzero$. It follows that $\by(\bu)\geq\underline{\by}$ and so $\bef(\bu) = \bzero$ from~\eqref{eq:fdef}. Upon initializing the proposed control scheme at $\bu$ and $\blambda$, \emph{Step 1} provides $\bef(\bu) = \bzero$; \emph{Step 2} yields $\hat{\bu}(1)=\bu$; \emph{Step 3} provides $\blambda(1)=\bzero$; and \emph{Step 4} that $\bu(1) = \bu$. Therefore, $(\bu,~\blambda)$ is an equilibrium for the proposed control steps.

To establish the reverse direction, we will prove the contrapositive statement, i.e., if $\bu$ does not belong to $\mcF$ or $\blambda\neq\bzero$, then $(\bu,\blambda)$ is not an equilibrium. We show the two cases separately using proof by contradiction. 

\emph{Case 1)} 
Suppose $\bu\notin\mcF$, yet there exists a $\blambda$ such that $(\bu,\blambda)$ is an equilibrium for the algorithm. Because $\bu$ is projected in its permissible range by~\eqref{eq:x-project}, its infeasibility means there exists a node $n\in\mcY$ such that:
\begin{equation}
   \underline{y}_n>y_n(\bx)
    \label{eq:yn<y}.
\end{equation}

\emph{Step 2} dictates $\hat{u}_n= u_n+\eta_{1,n}f_n(\bx) + \eta_{2,n}\sum_{\ell=1}^N A_{n,\ell}\lambda_\ell$.

Because $f_n(\bx)>0$, it follows that $\hat{u}_n > u_n$; however, since \emph{Step 4} sets $u_n = \min\{\hat{u}_n,\bar u_n\}$, it follows that ${u}_n=\bar u_n$ and $\hat{u}_n>\bar{u}_n$, so  $\lambda_n > 0$ from \eqref{eq:lambda-project}.

Consider a node $m$ neighbor of $n$ in $\mcG_c$ and compute
%
% \begin{equation*}%\label{eq:xm}
$\hat{u}_m = u_m + \eta_{1,m}f_m(\bu) + \eta_{2,m}\sum_{\ell=1}^N A_{m,\ell}\lambda_\ell$.
% \end{equation*}
%
Because $\lambda_n >0$ and $u_m = \min\{\hat{u}_m,\bar u_m\}$, it holds again that $\hat{u}_m > u_m= \bar{u}_m$, so $\lambda_m > 0$. 
Repeating this argument for all neighbors of neighbors of $n$, and so on until all nodes have been covered, one gets $\blambda > \bzero$ and:
\begin{equation}\label{eq:x=xbar}
    \bx = \bar \bx.
\end{equation}
However, because $\mcF$ is assumed to be non-empty, Assumption~\ref{as:1} implies that $\by(\bar{\bx}) \geq \by({\bx'})$ for $\underline{\bx}\leq\bx'\leq\bar{\bx}$; hence:
\begin{equation}\label{eq:max_eff}
    \by(\bar{\bx}) \geq \underline{\by}.
\end{equation}
In other words, applying the maximum control effort makes the constraints on observed outputs hold. 
This concludes the proof for Case 1 because~\eqref{eq:yn<y} and \eqref{eq:x=xbar} contradict~\eqref{eq:max_eff}.

\emph{Case 2)}: Suppose there exists a $\blambda\neq\bzero$ and $\bu \in \mcF$, such that $(\bu,\blambda)$ is an equilibrium for the proposed scheme. By plugging~\eqref{eq:xhat} into~\eqref{eq:lambda-project}, we can express $\blambda(t+1)$ as:
\begin{align*}%\label{eq:lambda_alt}
    \blambda(t+1)=\max \big\{&\bzero,\diag(\bfeta_3) \big(\bu(t)+\diag(\bfeta_1)\bef(\bu(t)) + \notag \\
    &\diag(\bfeta_2)\bA\blambda(t) - \bar{\bu}\big) \big\}.
\end{align*}
The feasibility assumption provides $\bef(\bu)=\bzero$, and the equilibrium condition yields:
\begin{align}\label{eq:lambdatrain}
\blambda&= \max \big\{\bzero,\diag(\bfeta_3) \big(\bu + \diag(\bfeta_2)\bA\blambda - \bar{\bu}\big) \big\} \notag \\
 %   & \leq \max \big\{\bzero, \diag(\bfeta_3) \diag(\bfeta_2)\bA\blambda \big\} \notag \\
    & \leq \diag(\bfeta_3) \diag(\bfeta_2)\bA\blambda
\end{align}
where the %first and second inequalities stem 
inequality stems from the fact that $\bu - \bar{\bu}\leq\bzero$ and $\diag(\bfeta_2)\diag(\bfeta_3)\bA\blambda(t)\geq\bzero$.
Invoking the norm inequality on~\eqref{eq:lambdatrain}, we get ${\|\blambda\|\leq\|\diag(\bfeta_2)\diag(\bfeta_3)\bA\|\|\blambda\|}$. Also, using the condition $\|\diag(\bfeta_2)\diag(\bfeta_3)\bA\|<1$ with $\blambda\neq\bzero$ yields ${\|\blambda\|<\|\blambda\|}$, which is a contradiction.

\bibliographystyle{IEEEtran}
\bibliography{myabrv,power}

\end{document}